\documentclass[conference]{IEEEtran}
\IEEEoverridecommandlockouts
\usepackage{times,algorithm,algpseudocode,amsthm}
\usepackage{url}
\usepackage{graphicx,epstopdf}
\usepackage{subfigure}
\usepackage{enumitem}	
\usepackage{xcolor}
\usepackage{amsmath}
\usepackage{listings}
\usepackage{changepage}
\usepackage{pifont}
\usepackage{multirow}
\usepackage{wrapfig}
\usepackage{color,soul}
\usepackage{comment}
\usepackage{fancyvrb}
\fvset{fontsize=\footnotesize}
\RecustomVerbatimEnvironment{verbatim}{Verbatim}{}

\newcommand{\Paragraph}[1]{\vskip 3pt\noindent\textbf{#1 }}

\newtheorem{theorem}{Theorem}

\newenvironment{myitem}{
\begin{itemize}
\setlength{\parskip}{0pt}
\setlength{\itemsep}{0pt}
\setlength{\partopsep}{0pt}
\setlength{\parskip}{0pt}
\setlength{\topsep}{0pt}
\setlength{\parsep}{0pt}}{\end{itemize}
}

\newcommand{\mj}[1]{\sethlcolor{yellow} \hl{mj: #1}} 

\newcommand{\both}{AOMG}
\newcommand{\Topk}{Few-k merging} 
\newcommand{\topk}{few-k merging} 

\newcommand{\Qt}{Quantile}
\newcommand{\Qts}{Quantiles}
\newcommand{\qt}{quantile}
\newcommand{\qts}{quantiles}

\newcommand{\IS}{\texttt{InitialState}}
\newcommand{\Ac}{\texttt{Accumulate}}

\newcommand{\Dac}{\texttt{Deaccumulate}}

\newcommand{\dacon}{deaccumulation}
\newcommand{\CR}{\texttt{ComputeResult}}

\newcommand{\bing}{Search}
\newcommand{\engine}{Trill}
\newcommand{\pingmesh}{Pingmesh}

\newcommand*{\affmark}[1][*]{\textsuperscript{#1}}

\begin{document}

%
\title{Approximate Quantiles for Datacenter Telemetry Monitoring}

%

%
%
%
%

\author{

\IEEEauthorblockN{Gangmuk Lim} \IEEEauthorblockA{\textit{UNIST}} \and

\IEEEauthorblockN{Mohamed Hassan} \IEEEauthorblockA{\textit{Oracle}} \and

\IEEEauthorblockN{Ze Jin} \IEEEauthorblockA{\textit{Facebook}} \and

\IEEEauthorblockN{Stavros Volos} \IEEEauthorblockA{\textit{Microsoft
Research}} \and

\IEEEauthorblockN{Myeongjae Jeon} \IEEEauthorblockA{\textit{UNIST and
Microsoft Research}}

}

\maketitle

\begin{abstract}

Datacenter systems require efficient troubleshooting and effective resource
scheduling so as to minimize downtimes and to efficiently utilize limited
resources. In doing so, datacenter operators employ streaming analytics for
collecting and processing datacenter telemetry over a temporal window. The
quantile operator is key to these systems as it can summarize the typical
and abnormal behavior of the monitored system. Computing quantiles in
real-time is resource-intensive as it requires processing hundreds of
millions of events in seconds while providing high accuracy.

We overcome these challenges in real-time quantile computation through workload-driven approximation, motivated by three insights in our study: (i) values are dominated by a set of
recurring small values, (ii) distribution of small values is consistent
across different time scales, and (iii) tail values are dominated by a
small set of large values. That is, we propose \both{}, an efficient and accurate
\qt{} approximation algorithm that capitalizes on these insights. \both{}
minimizes memory footprint of the quantile operator via compression and
frequency-based summarization of small values. While these summaries are
stored and processed at sub-window granularity for memory efficiency, they
can extend to compute quantiles on user-defined temporal windows. Low value
error for tail quantiles is achieved by retaining a few tail values per
sub-window. \both{} estimates quantiles with high throughput and less than
5\% relative value error across a wide range of use cases while
state-of-the-art algorithms either have a high relative value error
(9.3-137.0\%) or deliver lower throughput (15-92\%).

\end{abstract}

\section{Introduction}
\label{sec:Introduction} 
Stream analytic systems are key components in large-scale systems 
as they play pivotal roles in monitoring their status and responding to events in real 
time~\cite{Qian13,Chandramouli14,Kulkarni15,Lin16}. 
For instance, datacenter network~\cite{Guo15,googleavailability} and web search
~\cite{Dean13,Jeon16} monitoring systems collect response latencies of
servers to assess the health of underlying systems and/or to guide resource scheduling
decisions. These monitoring systems \emph{continuously} receive massive
amounts of data from a number of machines, perform computations over \emph{recent}
data as scoped by a \emph{temporal window}, and \emph{periodically} report
results, typically in seconds or minutes. Real-time computation of complex operations over
such data volumes requires hundreds of machines~\cite{Qian13,Lin16}, calling
for improvements in stream processing throughput~\cite{Miao17}.

The quantile operator lies in the heart of real-time monitoring systems
as they can determine the typical (0.5-\qt{} or median) or abnormal behavior
(0.99-\qt{}) of the monitored system. 
For instance, in network troubleshooting, the quality of network reachability 
can be measured via a static set of quantiles that are computed on the
round-trip times (RTTs) between datacenter servers~\cite{Guo15,Pekhimenko18}. 
In web search engines, a predefined set of quantiles are computed on query 
response times across clusters and are employed by load balancers so as to 
meet strict service-level agreements (SLA) on query latency~\cite{Dean13}. 
In such scenarios, highly accurate quantiles are required to reduce
any false-positive discoveries.

Accurate and real-time computation of quantiles is challenging as exact and low-latency
computation of \qts{} is resource-intensive and often infeasible. Unlike
aggregation operators (e.g., average) that are computed incrementally with
a small memory footprint, quantile computation requires storing and
processing the entire value distribution over the temporal window. In
real-world scenarios, such as Microsoft's network latency monitoring~\cite{Guo15}
and SLA monitoring of user-facing applications~\cite{Bodik10}, 
where million events arrive every second and the
temporal window can be in the order of minutes, the memory and compute
requirements for exact and low-latency quantile computation are massive.
As a result, approximate \qts{} are often acceptable if they can be computed
with a fraction of resources needed for the exact solution.


In this paper, we uncover opportunities in approximate \qts{} by
characterizing real-world workloads from production datacenters. 
Our study shows that these workloads
have many recurring values and can have a substantial skew. For instance, in
our datacenter network latency dataset~\cite{Guo15,Pekhimenko18}, 
only $0.08$\% of the numerical values in a one-hour temporal window are \emph{unique}.
While most latencies are small and
concentrated, with more than 90\% taking below 1.25~ms, a few latencies are
very large and heavy-tailed, taking up to 75~ms. When studying the
distributions across different time scales, we also find that the
distribution of small values is self-similar. This is not
surprising because datacenter networks work well and function consistently
most of the time, resulting in similar latencies and distributions.
Our findings corroborate prior work which reveals that high redundancy in streaming data
is common in a wide variety of scenarios in datacenters as well as in Internet of Things (IoT)~\cite{Pekhimenko18}.

We propose \both{}, 
a new and practical approach on approximate \qts{} customized for
large-scale monitoring systems in datacenters. 
\both{} leverages the observation that
the quantiles to be monitored are fixed throughout temporal windows.
\both{} takes into account the underlying data distribution for these
quantiles and proposes different approaches for computing non-high \qts{} and
high \qts{}. Each approach capitalizes on the insights from our workload
characterization, delivering efficient and highly-accurate computation of quantiles
as follows:

(1) \textbf{Non-high quantiles.} Based on the observation that distribution
of small values (i.e., the ones used in computation of non-high quantiles)
is self-similar, \both{} first performs quantile computation at the
granularity of sub-windows (i.e., smaller window than the temporal window
defined by the user). The quantile for a given temporal window is then computed by
averaging the quantiles of all sub-windows falling within the temporal
window. During the sub-window computation, 
\both{} significantly reduces the memory consumption by capitalizing the high data
redundancy such that it maintains the frequency distribution of
in-flight data~(i.e., \{value, count\}) instead of the entire value distribution.
This optimization enables \both{} to store only a small collection of unique values
associated with the in-flight sub-window and the quantiles of a few preceding sub-windows,
resulting in higher throughput due to smaller memory footprint.

(2) \textbf{High quantiles.} \both{} explicitly records tail values to better
approximate the high \qts{} (e.g., 0.999-\qt{}). Typically, these values are
infrequent and can be stored efficiently. Our technique, called
\emph{\topk{}}, carefully chooses which and how many tail values are stored
based on the query parameters (i.e., window size and period) as well as
observed data distributions. We highlight two scenarios where \topk{} is
needed: (i) \emph{Statistical inefficiency}. When the sub-window contains too
few data points, the high \qts{} in each sub-window are not statistically
robust as they are impacted by too few values. For instance, if a sub-window
has $1000$ elements, the 0.999-\qt{} is decided only by the two largest
elements;
and (ii) \emph{Bursty tail}. If the distribution of tail values
is highly non-uniform across sub-windows due to bursty outlier values,
their impact on the overall \qts{} is
not reflected well by the \qts{} of their corresponding sub-windows. We discuss
how to merge few-k values to produce an
answer, how to manage space budget, how to detect bursty tail, etc.

Similar to prior work on approximate \qts{} on the recent 
temporal window model~\cite{Lin04,Arasu04,Luo16,Gan18}, \both{} reduces memory
consumption as fewer values are retained during the temporal window.
Prior work, however, seeks to steer the rank error of approximate \qt{} such that
the exact rank and and the approximate rank are within a small distance;
rank $r$ is the $r$-th largest value in the window.
This paper introduces a new metric, namely value error, for the first time which is
relative error in value produced by approximate \qt{} as compared to
the exact value. Several datacenter monitoring scenarios require approximate \qts{}
to achieve low value errors instead of low rank errors since they use the reported numbers
directly for latency pattern
visualization~\cite{Guo15,Pekhimenko18}, identifying performance crisis by comparing with
a number of threshold values~\cite{Bodik10},
resource scheduling based on request response times~\cite{Jalaparti13}, etc. 
Our study shows that rank error based methods bring out low value errors for
non-high \qts{} (e.g., 0.5- to 0.9-\qt{}), but they frequently fail to
achieve low value errors for higher \qts{} (e.g., 0.999-\qt{}).

We have implemented \both{} in \engine{} streaming analytics
engine~\cite{Chandramouli14} and evaluated 
\both{} using both real-world and synthetic workloads. 
Our experiments show
that, relative to computing the exact \qts{} (default in \engine{}), \both{} offers
$2.4$-$4.8\times$ higher throughput for the small temporal window
that includes 10K elements;
the throughput is up to $~62\times$ higher when the window includes 1M elements.
Moreover, compared to prior approximate \qt{} algorithms~\cite{Lin04,Arasu04,Luo16,Gan18}
that were built atop the rank-error approximation metric,
\both{} lowers space usage by around $5$-$20$ times, and the
average relative value error for different \qts{} all falls below $5\%$. In
comparison, prior algorithms either incur high relative value errors
($9.3-137.0$\%) for high \qts{} and thus higher rank errors
or result in lower throughput~($15-92$\% on a
sliding window of 100K elements that includes 10 sub-windows).

We have deployed \both{} in our streaming network monitoring
system in production datacenters~\cite{Guo15}. 
Our insights from production-workload characterization lead to
mechanisms and algorithms that deliver better accuracy and 
higher throughput than the state-of-the-art generic solutions that are
workload-agnostic, targeting wider workloads (that often are not found in practice)
at the cost of accuracy and efficiency.
In summary, this paper has the following contributions:
\begin{myitem}
\item We define the problem of approximate computing of \qts{} with low value
    error (rather than rank error) and present a solution, \both{}. To the
    best of our knowledge, this is the first attempt to tackle this problem.
\item We design mechanisms to reduce memory consumption
    through space-efficient summaries and frequency distribution,
    enabling high throughput in the presence of huge data streams and large temporal windows.
\item We implement \both{} in a streaming engine and show its practicality
    using real-world use cases.
\end{myitem}

\section{Datacenter Monitoring Workloads}
\label{sec:Motivations}

In this section, we introduce real-world scenarios of using \qt{} computation
for datacenter telemetry monitoring in detail and justify the use of value errors as
evaluation metrics instead of rank errors.
Continuous telemetry monitoring is indispensable for network health
dashboard, incident management, troubleshooting, and numerous scheduling
decisions in datacenters~\cite{Guo15,googleavailability,Bodik10,Jalaparti13,Dean13}.
Next, we discuss important real-world examples illustrating the context.

\Paragraph{(1) \pingmesh{} network monitoring system~\cite{Guo15}:}
Servers in Microsoft datacenters probe
each other to measure network latency between servers pairs. From these collected measurements,
a set of metrics including the (0.5, 0.9, 0.99, 0.999) \qts{}
are calculated to capture various network latency issues, such as sudden increases in
high \qt{} latency compared to the median latency (i.e., SLA issue) or the identification of
the podset or spine switch failure (i.e., HW issue).
For instance, if servers in a certain podset all suffer from high network latency (e.g.,
$>$4~ms)
for the 0.99-\qt{} for both inbound and outbound traffic, \pingmesh{} informs that 
there is a HW issue related to the podset switch.

\Paragraph{(2) Datacenter fingerprinting~\cite{Bodik10}:} 
To keep high availability and responsiveness of datacenter servers,
the datacenter operator monitors key end-to-end performance indicators (KPIs)
such as response latency and request throughput for each individual server.
The operator assigns each KPI an SLA threshold
to declare a performance crisis if a certain fraction of machines violate any KPI SLA
for a particular time interval. For instance, an SLA might require the server
interactive response time to be below a certain threshold value for
99.9\% of all requests (i.e., 0.999-\qt{}) in the past 15 minutes~\cite{Bodik10}.

\Paragraph{(3) Resource scheduling:} 
Numerous datacenter services have latency targets for a predefined set of \qts{}.
Given the computed set of \qts{} over a time period, the resource scheduling logic can estimate and provide resources that are \emph{sufficient} to achieve the latency targets.
Various optimizations are possible via such resource scheduling.
For instance, for services with many sequential stages, there can be a global
objective function designed to minimize higher percentiles of the end-to-end
latency~\cite{Jalaparti13}. Per stage, there can be a local objective function
designed to save costs by auto-scaling container sizes---i.e., meeting the per-stage latency targets with minimum resources~\cite{Das16}.

\subsection{Approximating \qts{} with low value errors.}

Due to the large scale of datacenters, comprising hundreds of
thousands servers, a telemetry monitoring system may process
tens of TB data per day or millions of events every second~\cite{Guo15,M3}.
Furthermore, systems such as health dashboard tracking and network
troubleshooting have two requirements that are often in collision: (i) low-latency computations to reduce the time-to-detection (TTD)
time, and (ii) highly accurate results to reduce false discoveries.
We satisfy these requirements through approximate computation of \qts{}
that provides low-latency, yet high-accuracy results.

Our problem space calls that the computation produces \qts{} with
``small relative error in value'' (\textit{i.e.}, close to the true values).
This is because the estimates of \qts{} are frequently compared to the predefined 
threshold or target values that the datacenter operator sets for the \qts{},
as highlighted in the aforementioned examples. Unfortunately, we find that 
this requirement has not been targeted by prior work~\cite{Arasu04,Luo16,Gan18}. Instead, rank error has been widely used by prior work as an approximation metric. A careful study of approximate results generated by minimizing these different metrics uncovers that a low rank error does not necessarily lead to a low value error.

For instance, 
consider a data stream of size $N$ with its elements sorted as $\{e_1, \cdots, e_N\}$ in 
increasing order.
The $\phi$-\qt{} ($0<\phi \leq 1$) is the ${\lceil \phi\\N \rceil}$-th element $e_{\lceil \phi\\N \rceil}$ with rank $r = {\lceil \phi\\N \rceil}$. For a given rank $r$ and data size $N$, prior work focuses on delivering an approximate \qt{} within a deterministic rank-error bound called $\epsilon$-approximation, \textit{i.e.}, the rank of approximate \qt{} is within the range $[r-\epsilon N, r+\epsilon N]$. Assume $\epsilon = 0.02$ and a window size of 100K elements, the rank error is bounded by $\epsilon N = 0.02 \times 100K = 2K$, thereby resulting in the rank interval $[r-2K, r+2K]$, where $r-2K$ and $r+2K$ are the lower and upper bounds respectively. If data is highly variant, a decent rank interval can turn into a large gap between the returned value and the exact value.

\begin{table}[!t]
\begin{adjustwidth}{-1in}{-1in}
\centering
\scriptsize
{
\begin{tabular}{c|ccc||ccc}
\hline
\multirow{2}{*}{\bf Datasets} & \multicolumn{3}{c||}{\bf 0.50-\qt{}} & \multicolumn{3}{c}{\bf 0.99-\qt{}}\\
\cline{2-7}
 & Exact & -2K & +2K & Exact & -2K & +2K\\\hline
\pingmesh{}~\cite{Guo15} & 798 & 781 & 814 & 1,874 & 1,516 & 74,265 \\
\bing{}~\cite{Jeon16} & 3,574 & 3,384 & 3,780 & 200,109 & 154,156 & 263,771 \\
TaxiTrip~\cite{NYCTaxi} & 1,570 & 1,500 & 1,620 & 19,610 & 17,000 & 74,700
\\\hline
\end{tabular}
}
\end{adjustwidth}
\caption{Data variability and subsequent impact of a rank distance.}
\label{tab:skew}
\vspace{-0.3in}
\end{table}


Table~\ref{tab:skew} shows variability in real-world datasets under our study, and the consequential effect of
having a fixed rank interval. For each dataset, we measure exact values at the 0.50-\qt{} and the
0.99-\qt{}, and the values observed with the rank interval of $2K$ (i.e., $\epsilon = 0.02$).
The gap between these two \qts{} is large for each dataset, thus resulting in
the same rank interval that delivers dramatically different value distances
according to underlying data distributions.
For example, \pingmesh{} at the 0.50-\qt{} ($r = 50K$) is 798~us, and its rank distance +2K (\textit{i.e.}, $r + 2K = 52K$) sits in 814~us, which is only 2\% relative value error.
In contrast, the latency at the 0.99-\qt{} ($r = 99K$) is 1,874~us, and the same rank distance (\textit{i.e.}, $r + 2K = 101K > 100K$) sits in the largest latency at 74,265~us, which is 39 times larger.
Therefore, when designing a \qt{} approximation algorithm,
unlike prior work, we must take into account the underlying
data distribution and its influence on value errors.


\section{\Qt{} Processing Model}
\label{sec:Preliminaries}

In this section, we introduce our streaming query execution model and define the problem of
our \qt{} approximation.

\Paragraph{Streaming model.} A data stream is a continuous sequence of data
elements that arrive timely. Each element $e$ has its value associated with a
timestamp $t$ that captures when $e$ occured.
A streaming query defines a \emph{window} to specify the elements to
consider in the query evaluation. For example, we can have a window that includes
the latest $N$ elements seen so far, where $N$ is the \emph{window
size}. Due to the continuous arrival of data, a window requires a
\emph{period} to determine how frequently the query must be evaluated.
For example, a query can
process the recent $N$ elements periodically upon every insertion of new $K$
elements, where $K$ is the \emph{window period}.
Windows could be defined by time parameters, e.g., evaluate the query every
one minute (window period) for the elements seen last one hour (window size).

This paper mainly considers two types of windowing models~\cite{Chandramouli14}
which are frequently used in real-time monitoring of recent events:
(1) \emph{Tumbling Window}, where window size is \emph{equal} to window
period, and (2) \emph{Sliding Window}, where window size is \emph{larger}
than window period. As window size and period are the same in tumbling
window, there is no overlap between data elements considered by two
successive query evaluations.
In contrast, a sliding-window
query overlaps between successive query evaluations, thus allowing
elements to be reused in continuous windows.

\Paragraph{Incremental evaluation.} Incremental
evaluation~\cite{Chandramouli14} supports a unified design logic to
efficient implementation of window-based queries.
The basic idea is keeping state $S$ for
query $Q$ to evaluate, so that state $S$ is updated while new elements are
inserted or old elements are deleted. State $S$ is typically smaller in size
than the data covered by a window, thus making use of resource efficiently.
Further, when computing the query result, using $S$ directly is typically
faster compared to a naive, stateless way that accumulates all elements at
the moment of evaluating the query. To implement an incremental operator,
developers should define the following functions~\cite{Chandramouli14}:
\noindent
\begin{myitem}
    \item \texttt{InitialState:() => S}: Returns an initial state \texttt{S}.
    \item \texttt{Accumulate:(S, E) => S}: Updates an old state \texttt{S} with newly arrived event \texttt{E}.
    \item \texttt{Deaccumulate:(S, E) => S}: Updates an old state \texttt{S} upon the expiration of event \texttt{E}.
    \item \texttt{ComputeResult:S => R}: Computes the result \texttt{R} from the current state \texttt{S}.
\end{myitem}

\noindent For example, the following illustrates how to write an average
operator using the incremental evaluation functions:

\begin{lstlisting}[frame=single]
InitialState: () => S = {Count : 0, Sum : 0}
Accumulate: (S, E) => {S.Count + 1, S.Sum + E.Value}
Deaccumulate: (S, E) => {S.Count - 1, S.Sum - E.Value}
ComputeResult: S => S.Sum / S.Count
\end{lstlisting}

Incremental evaluation on sliding windows tends to be slower to execute than
that on tumbling windows. The primary reason is that the tumbling-window
query is implemented with a smaller set of functions without \Dac{}.
In this case, the query accumulates all data of a period on an initialized
state, computes a result, and simply discards the state. In contrast, a
sliding-window query must invoke \Dac{} for every element to be deleted from
the current window.

\Paragraph{Approximation error.}
We propose \emph{value error} as an important accuracy metric for
approximate \qt{}. To approximate the $\phi$-\qt{}, we must
define an estimator $y_\phi: R^N \rightarrow R$ such that the estimated \qt{}
$y_\phi(e_1, \cdots, e_N)$ is expected to be
close to the truth, which is the exact \qt{} value $e_{\lceil \phi\\N
\rceil}$. This results in the absolute value error $|y_{\phi}(e_1, \cdots, e_N) - e_{\lceil {\phi}N \rceil}|$ between the estimate and
the truth. The aggregation estimator in use will depend
on the associated $\phi$ and the underlying distribution density of the data.


\section{Algorithm for Non-high Quantiles}
\label{sec:ApproximateEvaluation}

First, we present an algorithm that effectively deals with non-high \qts{}
with high underlying distribution density, and its cost and error analysis.
We also illustrate scenarios, where the algorithm alone is insufficient, to
motivate techniques introduced in Section~\ref{sec:FewkMerging}.

\subsection{Algorithm Overview}
The key idea is to partition a window into \emph{sub-windows}, and leverage
the results from sub-window computations to give an approximate answer of the
$\phi$-\qt{} for the entire window. Sub-windows are created following
windowing semantics in use, by which the size of each sub-window is aligned
with window period. These sub-windows follow the timestamp order of data
elements, i.e., sub-windows containing older data elements are generated
earlier than those containing newer ones. For each sub-window, we maintain a
small-size \emph{summary} instead of keeping all data elements. At the window
level, there is an aggregation function that merges sub-window summaries to
approximate the exact $\phi$-\qt{} answer.

\begin{figure}[t]
\centering
\includegraphics[width=3.4in]
{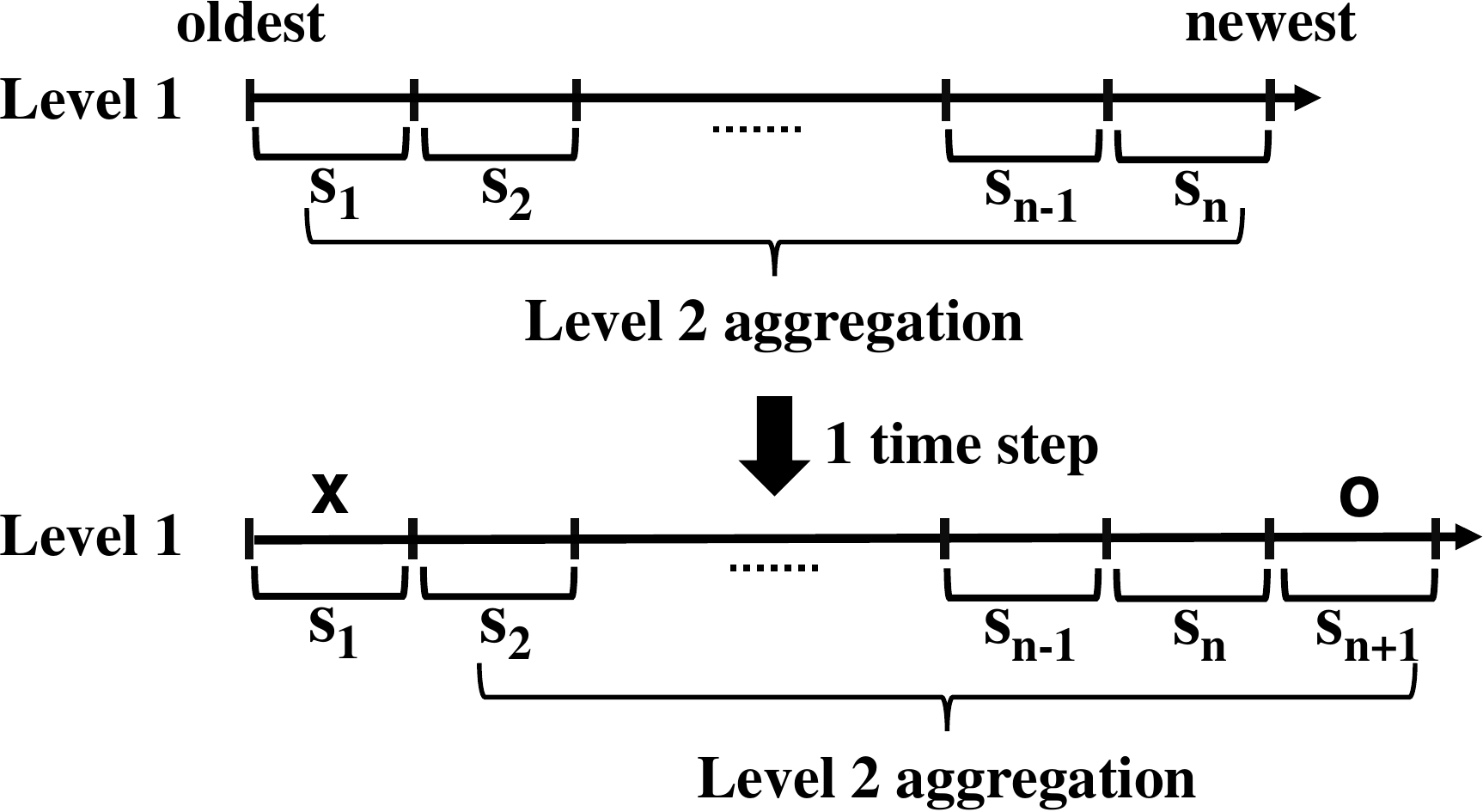}
\caption{Sliding window processing in \both{}.}
\label{fig:TwoLevel}
\vspace{-0.2in}
\end{figure}

More formally, assume a sliding window divided into $n$ sub-windows, where
each sub-window includes $m$ data points. If the $i$-th sub-window has a
sequence of data $x_i = \{ x_{i, 1}, \cdots, x_{i, j} : 1 \leq j \leq m\}$,
we observe the whole data $x = \{x_{i, j} : 1 \leq i \leq n, 1 \leq j \leq
m\}$ in the sliding window. Here, the sample $\phi$-\qt{} of the sliding
window is denoted by $y_e = \phi$-$\qt{}(x)$, which is the exact result to
approximate. \both{} estimates $y_e$ through \emph{two-level hierarchical
processing}, as presented in Figure~\ref{fig:TwoLevel}. $\textbf{Level 1}$
computes the exact $\phi$-\qt{} of each sub-window. The exact $\phi$-\qt{} of
the $i$-th sub-window is denoted as $y_i = \phi$-$\qt{}(x_i)$, which becomes
the summary $s_i$ of the corresponding sub-window. $\textbf{Level 2}$
aggregates the summaries of all sub-windows to estimate the exact
$\phi$-\qt{} $y_e$. \both{} uses mean as an aggregation function guided by
the Central Limit Theorem~\cite{stuart1994kendall, tierney1983space} and
obtains the aggregated $\phi$-\qt{} of the sliding window denoted by $y_a
=\frac{1}{n} \sum_{i = 1}^n y_i$.

When the window slides after one time step, as the figure shows \both{}
discards the oldest summary $s_1$, and adds $s_{n+1}$ for the new sub-window,
thereby forming a new bag of summaries $\{s_2, \cdots, s_{n+1}\}$.

In principle, \both{} is a hybrid approach that combines tumbling windows
into the original sliding-window \qt{} computation, delivering improved
performance. Specifically, while creating a summary \textbf{Level 1} runs a
tumbling window, which avoids \dacon{}. Once a sub-window completes, all
values are discarded and the summary only remains.
\textbf{Level 2} operates a sliding window on summaries, requiring \dacon{}.
However, since a summary contains only a few entries associated with the
specified \qt{}s, 
it can perform fast \dacon{} as well as fast accumulation.

\Paragraph{Creating a new sub-window summary.}
During Level~1 process, we exploit opportunities for volume reduction by data
redundancy. During the sub-window processing, in-flight data are maintained
in a compressed state of $\{(e_1, f_1), \cdots, (e_n, f_n)\}$, where $f_i$ is
the frequency of element $e_i$. A critical property here is that each $e_i$
is unknown until we observe it from new incoming data. To efficiently insert
a new element in the state and continuously keep it ordered, we use the
red-black tree where the element attribute $e_i$ acts as the
key to sort nodes in the tree. This also avoids the sorting cost when
computing the exact \qts{} at the end of Level~1 processing for the in-flight
sub-window.

The logic to manage the red-black tree is
sketched in Algorithm~\ref{alg:exact}. \IS{} and \Ac{} are self-explanatory,
and we explain \CR{} in details. At the result computation in \CR{}, the tree
already has sorted the sub-window's elements. Thus, the computation does an
in-order traversal of the tree while using the frequency attribute to count
the position of each unique element. As the total number of elements is
maintained in $state.Count$, it is straightforward to know the percentage of
elements below the current element during the in-order traversal. A query may
ask for multiple \qts{} at a time. In this case, \CR{} evaluates the \qts{}
in a single pass with the smallest one searched first during the same
in-order traversal.

Lastly, to increase data redundancy, some insignificant low-order digits of
streamed values can be zeroed out. For example, in our network monitoring system,
we consider only the three most
significant digits of the original value, which ensures the quantized value
within less than 1\% relative error.

\Paragraph{Aggregating sub-window summaries.}
The logic for aggregating all sub-window summaries is almost identical to the
incremental evaluation for the average introduced in
Section~\ref{sec:Preliminaries}. The only distinction is that if the number of 
\qts{} to answer is $l$, there are $l$ instances of the average's state
(\textit{i.e.}, sum and count). The accumulation and \dacon{} handlers update
these states to compute the average of each \qt{} separately.

\subsection{Algorithm Analysis}
\label{sec:Analysis}

\Paragraph{Space complexity.} The space complexity for our approximate
algorithm is $l(N/P)+O(P)$. For summaries of $l$ independent \qts{} $\{\phi_i
: 1 \leq i \leq l\}$, we need $l(N/P)$ space, where $N$ and $P$ are the
window size and the sub-window size, respectively. There is at most one
sub-window under construction, for which we maintain a sorted tree of
in-flight data. Its space usage is $O(P)$ which can range from 1 to $P$
depending on the degree of duplicates in the data. In one extreme case,
all elements have the same value, so $O(P) = 1$, and in another extreme case,
there is no duplicate at all, so $O(P) = P$. This spectrum allows us to 
reduce space usage significantly if there is high data redundancy in the workload.

\Paragraph{Time complexity.} The degree of data redundancy is a factor that also
reduces the theoretical time cost in Level 1 stage. For the sub-window of
size $P$ with $n$ unique elements, the \Ac{} cost is $O(log~n)$, which falls
along a continuum between $O(log~1)$ and $O(log~P)$, again depending on the
degree of data duplicates. Likewise, the complexity of \CR{}
is $O(n)$ irrespective of the number of \qts{} to search. Level 2 stage in
\both{} runs extremely fast with a static cost: each of $l$ specified \qts{}
needs two add operations for \Ac{} or \Dac{}, and one division operation for
\CR{}. Section~\ref{sec:Sensitivity} shows an experiment on how higher data
redundancy leads on higher throughput.


{
\begin{algorithm}[t]
\caption{Incremental computation for Level 1}
\label{alg:exact}
\footnotesize
\begin{algorithmic}[1]
\Procedure{\IS{}: }{}
\State \Return $state$ \Comment{new red-black tree}
\EndProcedure

\vspace{-10pt}
\Statex
\Procedure{\Ac{}: }{state, input}
\If{$state.ContainsKey(input) = false$}
    \State $state.CreateKey(input)$
\EndIf
\State $state.IncrementValue(input)$
\State $state.Count := state.Count + 1$
\EndProcedure

\vspace{-10pt}
\Statex
\Procedure{\CR{}: }{state}
\State $runningCount := 0$
\State $result_i := 0,~~i = 1,\dots,l$ \Comment{$l$ \qts{} to answer}
  \State $\phi.Sort()$ \Comment{\qts{} in non-decreasing order}
\State $i := 1$
\State $rank := \lceil \phi_i \times state.Count \rceil$ \Comment{rank of the first \qt{}}
\For{$t := node~in~the~inorder~traversial~of~state$}
    \State $runningCount := runningCount + t.Value$
    \While{$runningCount \geq rank$}
        \State $result_i = t.Key$
        \If{$i = l$}
            \State \Return $result$
        \EndIf
        \State $i := i + 1$
        \State $rank := \lceil \phi_i \times state.Count \rceil$ \Comment{rank of the next \qt{}}
    \EndWhile
\EndFor
\EndProcedure
\end{algorithmic}
\end{algorithm}
}

\Paragraph{Error bound.} Summary-driven aggregation is designed based on
self-similarity of data distribution for non-high \qts{} with high underlying
distribution density, as described in Section~\ref{sec:Introduction}. We
present some initial results on error bound analysis.
Our theorem assumes several conditions. (1) We consider the $\phi$-\qt{} of
each sub-window as a random variable before we actually get the data.
Similarly, the $\phi$-\qt{} of the sliding window is also a random variable
as long as the data is not given. (2) The target $\phi$-\qts{} across
sub-windows are independent and identically distributed (\emph{i.i.d.}). (3)
Data in the window have continuous distribution.

Now, we derive a statistical guarantee to have the aggregated estimate $y_a$
close to the exact $\phi$-\qt{} $y_e$ using the Central Limit Theorem (see
Theorem~\ref{unique} in Appendix for details). To illustrate how to interpret the results
in Theorem~\ref{unique}, suppose we obtain the aggregated estimate $y_a$ and
the error bound $e_b$ over the data stream of a window. Since the exact \qt{}
is $y_e$, we use $y_a$ to approximate $y_e$, and evaluate how close they are
using $e_b$. Essentially, we claim that $y_e \in [y_a - e_b, y_a + e_b]$ with
high confidence (e.g., probability 95\%).

Our probabilistic error bound depends on the density of underlying data
distribution for the specified $\phi$-\qt{}. Recall
the workloads presented in Section~\ref{sec:Motivations} where the density only decays in the tail,
making the density at the 0.5-\qt{} (median) much larger than that at the
0.999-\qt{}. In this case, the error bound is expected to be much tighter for
non-high \qts{} than for high \qts{}. Narrower error bounds imply lower
estimation errors, otherwise the error bound is not informative. For a number
of tests we performed, including a query for
Table~\ref{tab:competing_policies}, we see that the observed value error is
much lower than the error bound $e_b$.


Finally, due to condition (2) assumed in our theorem, the dependence between
data is not supposed to collapse our method. In
Section~\ref{sec:Sensitivity}, we show that our aggregated estimator can be
effectively applied to non-\emph{i.i.d.} data in a diverse spectrum of data
dependence, with competitive accuracy compared to \emph{i.i.d.} data.

\subsection{Special Consideration for High Quantiles}
\label{sec:Limitations}

There are two complementary cases where achieving high accuracy for high
\qts{} needs special consideration.

\Paragraph{Statistical inefficiency.}
The inaccuracy of high \qts{} becomes more significant when there is lack of
data to accurately estimate the \qts{} in a sub-window. For example, in
Figure~\ref{fig:relative_errors}, the estimated error increases noticeably at
the 0.999-\qt{} if sliding windows use 1K elements in a period
(\textit{i.e.}, 1K period). In this case, the two largest elements are used
when computing the 0.999-\qt{} in each sub-window. This makes statistical
estimation not robust under the data distribution, and thus misleads the
approximation.


There are several remedies for this curse of high \qts{}. Users can change
the parameters of query window to operate larger sub-windows with more data
points. This provides better chance to observe data points in the tail,
allowing more precise estimates of high \qts{}. Another approach is to cache
a small proportion of raw data without changing windowing semantics, and use
it to directly compute accurate high \qts{} of the sliding window. 

\Paragraph{Bursty tail.}
When bursty tail happens, extremely large values are highly skewed in one
or a few sub-windows. In effect, they dominate the values to be observed across
the window for computing high \qts{}.

\section{\Topk{}}
\label{sec:FewkMerging}

\Topk{} in \both{} leverages \emph{raw data points} to handle large value
errors that could appear in high \qts{} as a result of statistical
inefficiency or bursty tail. During \topk{}, each sub-window collects $k$
data points among the largest values in its portion of streaming data and
uses the $k$ values to compute the target high \qt{} for the window.

\subsection{Challenges}
\label{sec:TopkChallenges}

\begin{figure}[t]
\centering
\includegraphics[width=3.2in]
{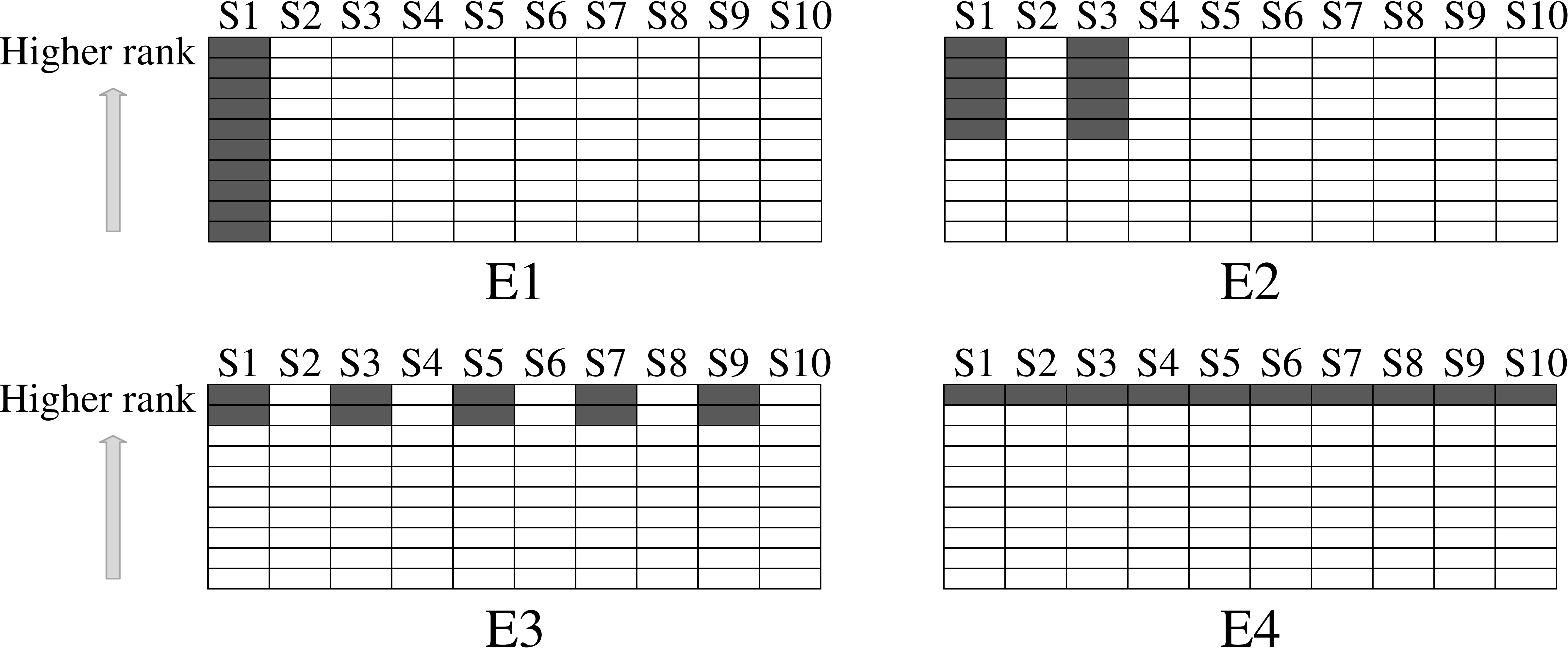}
\caption{Examples ($E1 - E4$) where the largest 10 values (colored in dark) appear differently at sub-windows ($S1 - S10$).}
\label{fig:windows}
\vspace{-0.2in}
\end{figure}

We begin
with discussing the issues that make collecting right $k$ values challenging.
Figure~\ref{fig:windows} exemplifies four patterns ($E1 - E4$),
where the largest 10 values (colored in dark) of the window are distributed
differently among sub-windows ($S1 - S10$). Assume the target high \qt{} can
be obtained precisely by exploiting these 10 values. Then, $E1$ indicates the
case of extremely bursty tail, where a single sub-window $S1$ includes all
the largest values, whereas $E4$ indicates the case that they are completely
evenly distributed across sub-windows. \Topk{} must enforce $k = 10$ for $E1$
to produce the exact answer. However, for $E4$, any $k \geq 1$ caters to
the precise \qt{}. Using $k = 1$ will sacrifice the accuracy for other cases,
with $E1$ performing the worst.
Driven by
this observation, our \emph{first challenge} is providing a solution to
handle both statistical inefficiency and bursty tail under diverse
patterns corresponding to the largest values over sub-windows.

Bursty tail occurs time-to-time as a result of
infrequent yet unpredictable abnormal behavior of the monitored system.
Therefore, we cannot assume that the distribution of the in-flight largest
values is known ahead or repeats over time. The distribution can only be
estimated once we observe the data. Thus our \emph{second challenge} is
building a robust mechanism to dynamically recognize the current burstiness
pattern. 


\subsection{Our Method}
\label{sec:MethodFewK}

One possible method is claiming enough data points to compute the exact
\qt{} regardless of statistical inefficiency and bursty traffic. Formally, to
guarantee the exact answer of $\phi$-\qt{} ($0<\phi \leq 1$) on the window
size of $N$, each sub-window must return $k = N(1-\phi)$ largest elements.
Then, the entire window will need to have space for $N(1-\phi)N/P$ elements
in total, where $P$ is the sub-window size. This approach could be costly if
the window size of a given sliding window query is significantly larger than
the window period. We thus consider \topk{} when the space is limited.

Let $B$ be the space budget, where $B$ is smaller than the space for the
exact $\phi$-\qt{}, \textit{i.e.}, $B < N(1 - \phi)N/P$. Given such $B$, we
assign each sub-window the same space budget $k$, where $B = k \cdot N/P$,
and $k < N(1 - \phi)$. Within each sub-window, $k$ will be further
partitioned into two parts, $k_t$ to address statistical inefficiency by
\emph{top-k merging} and $k_s$ to address bursty tail by \emph{sample-k
merging}, such that $k = k_t + k_s$.
We now explain how to use the given sub-window budgets, $k_t$ and $k_s$, to
handle statistical inefficiency and bursty tail.

\Paragraph{Top-k merging for statistical inefficiency.} When handling
statistical inefficiency, each sub-window caches its $k_t$ \emph{largest
values}, based on the observation that the global largest values required to
compute high \qts{} tend to be
\emph{widely spread} similar to $E4$ in Figure~\ref{fig:windows}.
These all $k_t$ largest values are then merged across the entire window
to answer $\phi$-\qt{}. For window size $N$, we draw the $N(1 - \phi)$th
largest value from the merged data to approximate the $\phi$-\qt{}.
Section~\ref{sec:FewkEval} shows that this method can indeed
trade-off small space consumption for high accuracy (\textit{i.e.},
low value error) in approximating high \qts{} in real-world scenarios.

\Paragraph{Sample-k for bursty traffic.} 
In bursty tail, the largest values in some sub-windows are relatively worse than
those in others.
Thus when coping with bursty tail, we do not differentiate the largest
values within each sub-window, unlike the top-k merging that considers
higher-rank values in each sub-window more important.

In the sample-k merging, each sub-window takes $k_s$ \emph{samples} from its $N(1 - \phi)$
largest values so as to capture the distribution of the largest values using
a smaller fraction. It takes an interval sampling which picks every $i$-th
element on the ranked values~\cite{Luo16}; e.g., for $i = 2$, we select all
even ranked values. The sampling interval is inversely proportional to the
allocated fraction $\alpha$ = $\frac{k_s}{N(1 - \phi)}$.
After merging all samples, the resulting $\phi$-\qt{} is obtained by
referring to the $\alpha N(1 - \phi)$th largest value to factor in data
reduction by sampling.

\subsection{Runtime Handling}
\label{sec:TrafficHandling}

The two proposed value merging techniques are triggered by different conditions
during the streaming query processing. The
use of top-k merging is decided by query semantics (e.g., windowing
model and target \qts{}). If there is any high \qt{} that suffers from statistical
inefficiency, the top-k merging for the \qt{} will be activated. In contrast,
sample-k merging is a standing pipeline and will be exploited 
anytime by ongoing bursty tail.
Currently, several decisions made for runtime handling are guided 
empirically or by parameters measured offline. Future work includes
integrating these processes entirely online.

\Paragraph{Enabling top-k merging.} For each $\phi$-\qt{}, we initiate
the top-k data merging process if $P (1 - \phi)$ is below the
threshold that decides the statistical inefficiency, for using sub-window of
size $P$. Otherwise, we directly exploit results estimated from our
approximation algorithm presented in Section~\ref{sec:ApproximateEvaluation}.
We set the threshold as 10 based on evaluating several monitoring workloads
we run in our system.

\Paragraph{Allocating space.} When deciding $k_t$ for top-k merging, we
estimate the number of data points that a window needs to compute an answer
for the target high \qt{}, e.g., $N(1-\phi)$ data points for $E1$
in Figure~\ref{fig:windows}. 
We use it to decide per sub-window data points, which we use as $k_t$. 
Note that we could assume a more conservative case such as
$E2$ in Figure~\ref{fig:windows}, then we increase $k_t$. 
\both{} assigns all the remaining budget for $k_s$.
$k_s$ is typically larger than $k_t$ because $k_t$ is based on a very small
portion of the largest data.
This is also to make sample-k merging take advantage of it,
since the value error in the estimation of high \qts{}
through sampling is sensitive to sampling rate due to low density of
underlying data.

\Paragraph{Selecting outcomes.} \both{} needs to decide \emph{when} to take
which outcome between the top-k merging and the sample-k merging at runtime.
Results from the sample-k
merging are prioritized if bursty tail is detected. Otherwise, \both{}
uses the results from the top-k merging for those high \qts{} that face
statistical inefficiency. To detect bursty tail, we identify if the
sampled largest values in the current sub-window are distributionally
different and stochastically larger than those in the adjacent former
sub-window.
We develop a simple way (similar to~\cite{Mann47}) to detect bursty tail as
follows. Suppose we
choose the sampled largest values $x_1, \cdots, x_n$ and $y_1, \cdots, y_n$
from two adjacent sub-windows. Then, those values are identified distributionally 
different if the statistics $\sum_{i=1}^n\sum_{j=1}^n \textrm{sign}(x_i - y_j)$ is
far from zero. \both{} decides the current traffic is bursty when it detects one or
a few sequential comparisons in the window that turn out to be distributionally
different. This metric is less affected by outliers, and does not perform any
strict ordering of the data. The sequential comparison between two adjacent
sub-windows may lose some information, but is efficient and fits the steaming
data flow.


\section{Experimental Evaluation}
\label{sec:ExperimentalEvaluation}

We implement \both{} along with competing policies using \engine{}
open-source streaming analytics engine~\cite{Chandramouli14}, and conduct
evaluation using both real-world and synthetic workloads.

\begin{table*}[t]
\begin{adjustwidth}{-1in}{-1in}
\centering
\footnotesize
{
\begin{tabular}{c|cccc|cccc||cc|cc||c}
\hline

& \multicolumn{12}{c||}{\bf Accuracy} & \multirow{2}{*}{\bf Space usage} \\
\cline{2-13}

\multirow{3}{*}{\bf Policy} & \multicolumn{8}{c||}{\bf \pingmesh{}} & \multicolumn{4}{c||}{\bf \bing{}} & \\
\cline{2-14}

& \multicolumn{4}{c|}{\bf Rank error ($e^\prime$)} &
\multicolumn{4}{c||}{\bf Value error (\%)} &
\multicolumn{2}{c|}{\bf Rank error ($e^\prime$)} &
\multicolumn{2}{c||}{\bf Value error (\%)} &
{\bf Observed}\\
\cline{2-13}

& Q0.5 & Q0.9 & Q0.99 & Q0.999 & Q0.5 & Q0.9 & Q0.99 & Q0.999 &
Q0.95 & Q0.999 & Q0.95 & Q0.999 &
{\bf (Analytical)} \\
\hline\hline

\both{} & .0016 & .0005 & .0002 & .0001 & 0.10 & 0.06 & 0.78 & 4.40 &
.0001 & .0001 & 0.34 & 0.02 & 3,340 (16,416) \\
CMQS & .0016 & .0018 & .0009 & .0007 & 0.31 & 0.26 & 1.78 & 28.47 &
.0010 & .0014 & 10.56 & 0.58 & 31,194 (33,504) \\
AM & .0013 & .0012 & .0009 & .0005 & 0.14 & 0.22 & 2.19 & 76.49 &
.0009 & .0003 & 7.15 & 0.43 & 36,253 (45,309) \\
Random & .0005 & .0006 & .0005 & .0007 & 0.05 & 0.12 & 1.11 & 137.03 &
.0006 & .0007 & 4.93 & 4.73 & 68,001 (45,611) \\
Moment & .0180 & .0017 & .0004 & .0002 & 0.98 & 0.28 & 0.76 & 9.30 &
.0008 & .0010 & 2.81 & 12.14 & 16,596 (NA) \\
\hline

\end{tabular}
}
\end{adjustwidth}
\caption{\label{tab:competing_policies}Accuracy and space usage of five approximation algorithms.}
\vspace{-0.2in}
\end{table*}

\subsection{Experimental Setup}
\label{sec:ExperimentalSetup}

\Paragraph{Machine setup and workload.} The machine for all experiments has
two 2.40 GHz 12-core Intel 64-bit Xeon processors with hyperthreading and
128GB of main memory, and runs Windows. Our evaluation is based on the
four datasets.
(1) \textbf{\pingmesh{}} dataset~\cite{Guo15} includes network
latency in microseconds (us), with each entry measuring round-trip time (RTT)
between any two servers in a Microsoft datacenter. (2) \textbf{\bing{}}
dataset~\cite{Jeon16} includes query response time of index serving node (ISN) at
Microsoft Bing. The response time is in microseconds (us), and measures from the time
the ISN receives the query to the time it responds to the user. (3) \textbf{Normal}
dataset includes integers randomly generated
from a normal distribution, with a mean of 1 million and a standard
deviation of 50 thousand.
(4) \textbf{Uniform} dataset includes integers randomly generated
from a uniform distribution ranging from 90 to 110.
Each real-world dataset contains 10 million entries, and
each synthetic dataset contains 1 billion entries.

Asides from these datasets, we also test \both{} using
non-i.i.d. datasets in Section~\ref{sec:Sensitivity} and a geospatial
IoT dataset in Section~\ref{sec:Discussion} to show
its efficacy on a variety of scenarios.




\Paragraph{Query.} We run the query Qmonitor written in LINQ on the
datasets to estimate 0.5, 0.9, 0.99, and 0.999-\qt{} (henceforth denoted by
Q0.5, Q0.9, Q0.99, and Q0.999):
\begin{verbatim}[fontsize=\small]
Qmonitor = Stream
 .Window(windowSize, period)
 .Aggregate(c => c.Quantile(0.5,0.9,0.99,0.999))
\end{verbatim}


\Paragraph{Policies in comparison.} We compare \both{} to the four strategies
that support both tumbling and sliding window models. (1) \textbf{Exact}
computes exact \qts{}. This extends Algorithm~\ref{alg:exact}
with a \dacon{} logic; the node representing the expired element
decrements its frequency by one, and is deleted from the tree if
the frequency becomes zero. This outperformed other methods for computing exact
\qts{}. (2) \textbf{CMQS}, Continuously Maintaining \Qt{}
Summaries~\cite{Lin04}, bounds rank errors of \qt{} approximation
deterministically; for a given rank $r$ and $N$ elements, its
$\epsilon$-approximation returns a value within the rank interval
$[r-\epsilon N, r+\epsilon N]$. (3) \textbf{AM} is another deterministic
algorithm with rank error bound designed by Arasu and Manku~\cite{Arasu04}.
(4) \textbf{Random} is a state of the art using sampling to bound rank error
with constant probabilities~\cite{Luo16}. (5) \textbf{Moment} Sketch is an
algorithm using mergeable moment-based quantile sketches to predict the
original data distribution from moment statistics summary~\cite{Gan18}.

\Paragraph{Metrics.} We use average relative error as the accuracy
metric, number of variables as the memory usage metric, and
million elements per second (M ev/s) processed by a single thread as
the throughput metric. The average relative error (in \%) is measured by
$\frac{1}{n}\sum_{i=1}^n\frac{|a_i - b_i|}{b_i}100$, where $a_i$ is estimated
value from approximation and $b_i$ is the exact value. As stream processing
continuously organizes a window and evaluates the query on it, the error is
averaged over $n$ query evaluations.

\subsection{Comparison to Competing Algorithms}
\label{sec:CompetingAlgo}

This section compares \both{} with the competing policies. We disable
\topk{} in \both{} until Section~\ref{sec:FewkEval} to show how our algorithm
in Section~\ref{sec:ApproximateEvaluation} alone works.

\Paragraph{Approximation error.} The accuracy column in
Table~\ref{tab:competing_policies} shows average value error and rank error
for a set of \qts{} when using 16K window period and 128K window size on
\pingmesh{} and \bing{} datasets.
For CMQS, AM and Random, we configure the error-bound parameter $\epsilon$ as
0.015 guided by its sensitivity to value error. For Moment, we set its $K$
parameter as 12 to be similar in error bounds. For a given \qt{} $\phi$, in
addition to average value error, we present its average rank error measured
by $e^\prime = \frac{1}{n}\sum_{i=1}^n|\frac{r - r^\prime_i}{N}|$, where $r$
is the exact rank of $\phi$, $r^\prime_i$ is the rank of the returned value
for $i$-th query evaluation, $n$ is the total number of query evaluations,
and $N$ is the window size.
$\epsilon = 0.015$ guarantees that none of $|\frac{r - r^\prime_i}{N}|$ is
larger than 0.015.

The results in Table~\ref{tab:competing_policies} show that CMQS, AM and
Random can all successfully bound errors by rank. The average rank errors
stay low within $\epsilon = 0.015$, with the largest error observed in
individual query evaluations across all the policies below 0.0125, which
confirms the effectiveness of the proposed methods. Moment's rank
error is also comparable. It is noteworthy that the rank error of values
returned by \both{} is comparable, with even slightly lower
across the \qts{}.

Comparing across the policies, \both{} outperforms others in value error,
especially for very high \qt{} in \pingmesh{}.
Moreover, we see that
a given rank error has very different influence on the value error across
different \qts{}. For example, comparing Q0.5 and Q0.999, the rank error in
Q0.999 is lower while its corresponding value error is adversely higher. This
is primarily because the \pingmesh{} workload exhibits high variability,
where as presented in Section~\ref{sec:Motivations} the value for the high
\qt{} is orders magnitude larger than the medium \qt{}.
As a result, small change in Q0.999 rank error ends up leading to
$2.1$-$31.1\times$ difference in value error.

Interestingly, such huge value errors appeared in Q0.999 are
less apparent with \bing{} workload as shown in
Table~\ref{tab:competing_policies}.
This is primarily because \bing{} ISN initiates query termination if a query
executes longer than the pre-defined response time SLA, e.g.,
200~ms. Those terminated queries are accumulated on Q0.9 and above.
However, looking at Q0.95, whose response time is typically smaller than
the SLA, we still see relatively large value errors in CMQS, AM, Random, and
Moment, whereas the error is only 0.34\% in \both{}.



\Paragraph{Space usage.} The space usage column in
Table~\ref{tab:competing_policies} presents the number of variables to store
in memory for each algorithm. The space usage is calculated from the
theoretical bound (Analytical) found in~\cite{Lin04,Arasu04,Luo16,Gan18},
as well as measured at runtime (Observed) while running the algorithm.
\both{} benefits from high data redundancy present in the \pingmesh{},
reducing memory usage from its analytical cost substantially.
Recall that our theoretical cost is $4(N/P)+O(P)$ (see
Section~\ref{sec:Analysis}) for the window size $N$ and the period size $P$.
For $O(P)$, the actual cost approaches to 1 from $P$ as we see more data
duplicates in the workload. This is how \both{} reduces memory consumption in
practice.

\begin{figure}[!t]
\centering
\includegraphics[width=2.4in]{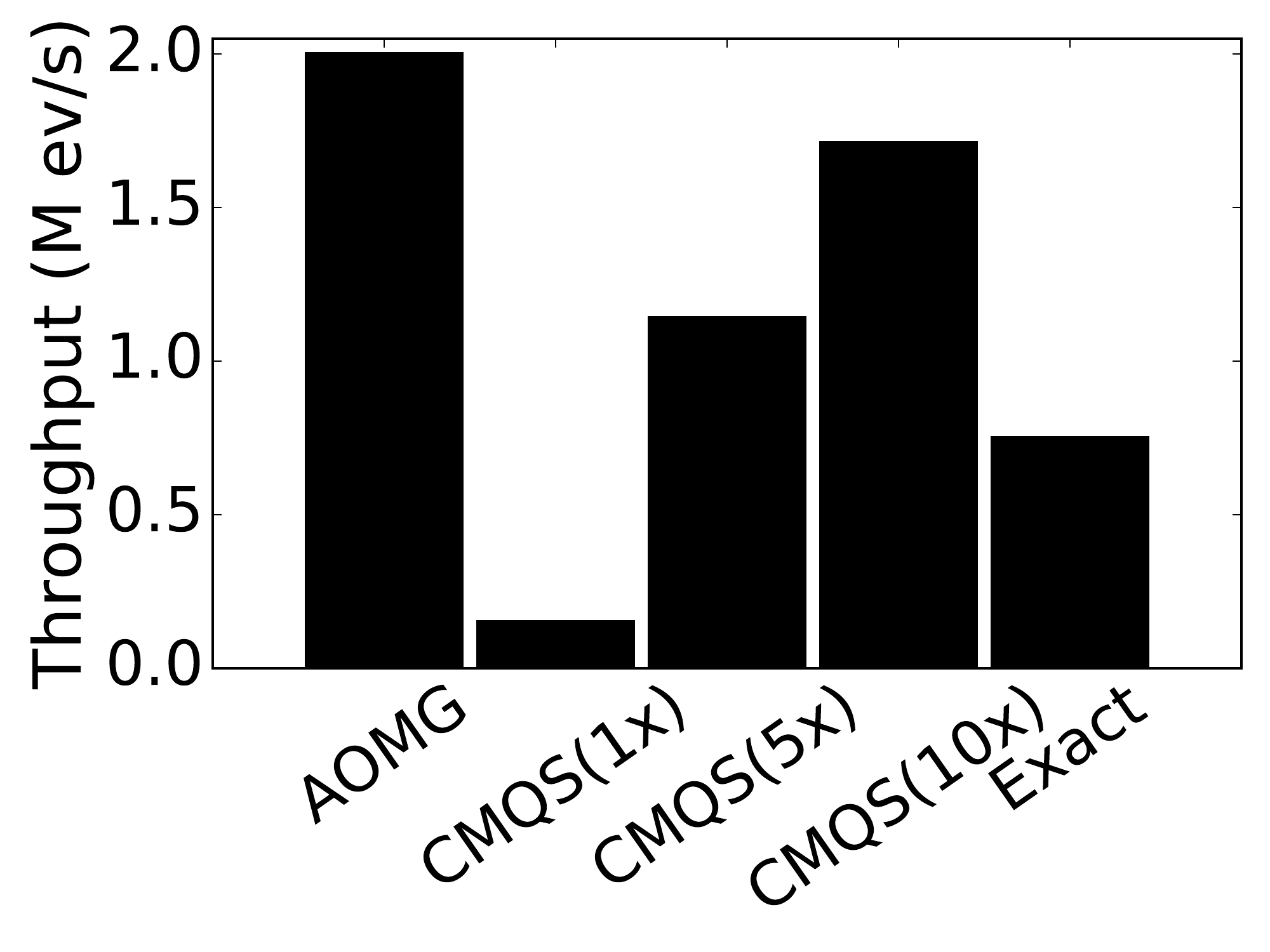}
\caption{\label{fig:gk_throughput} Throughput comparison.}
\vspace{-0.2in}
\end{figure}

\begin{figure*}[!t]
\begin{adjustwidth}{-1in}{-1in}
\center
\subfigure[\pingmesh{}]{\includegraphics[height=1.25in]{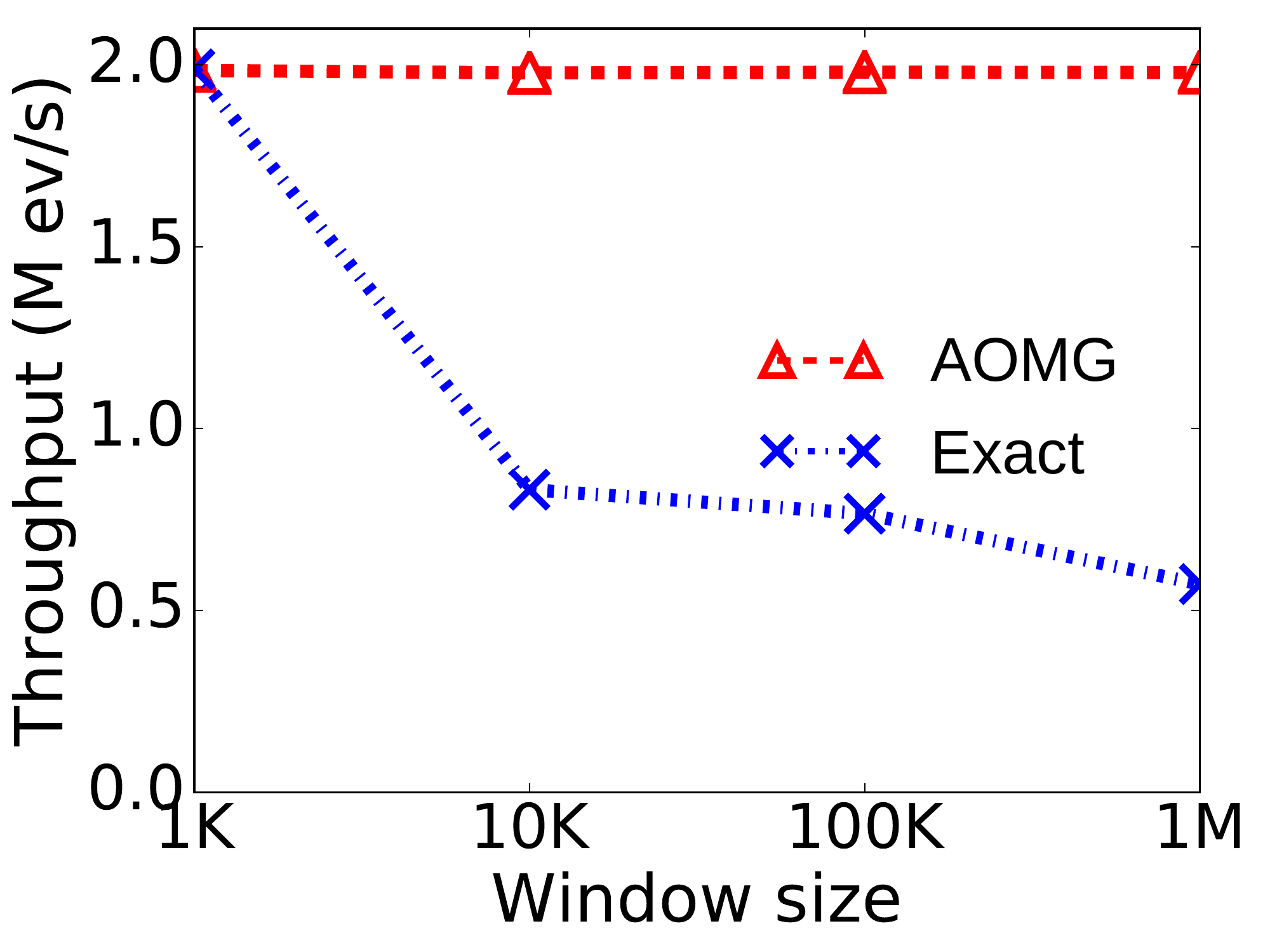}}
\subfigure[\bing{}]{\includegraphics[height=1.25in]{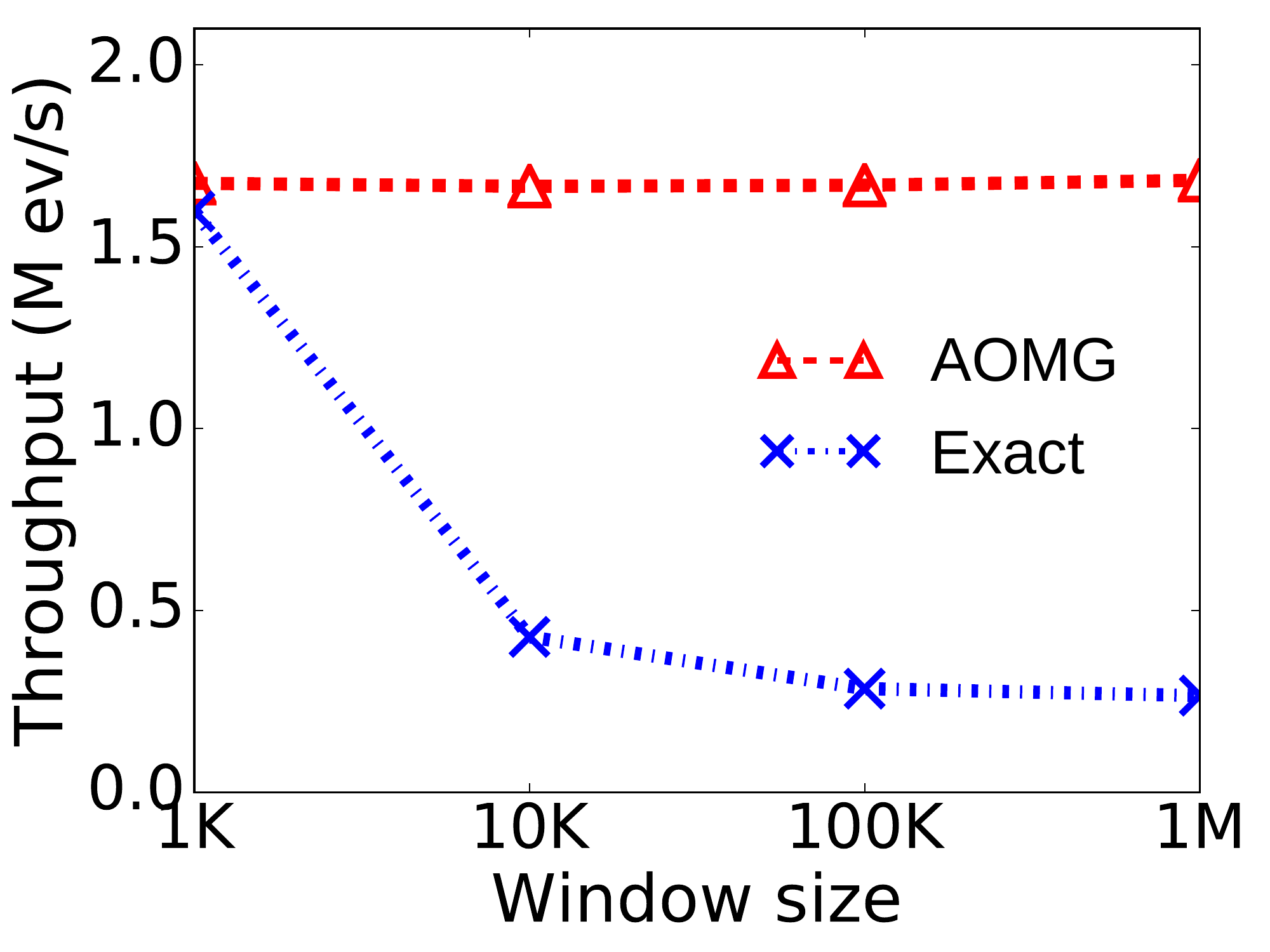}}
\subfigure[Normal]{\includegraphics[height=1.25in]{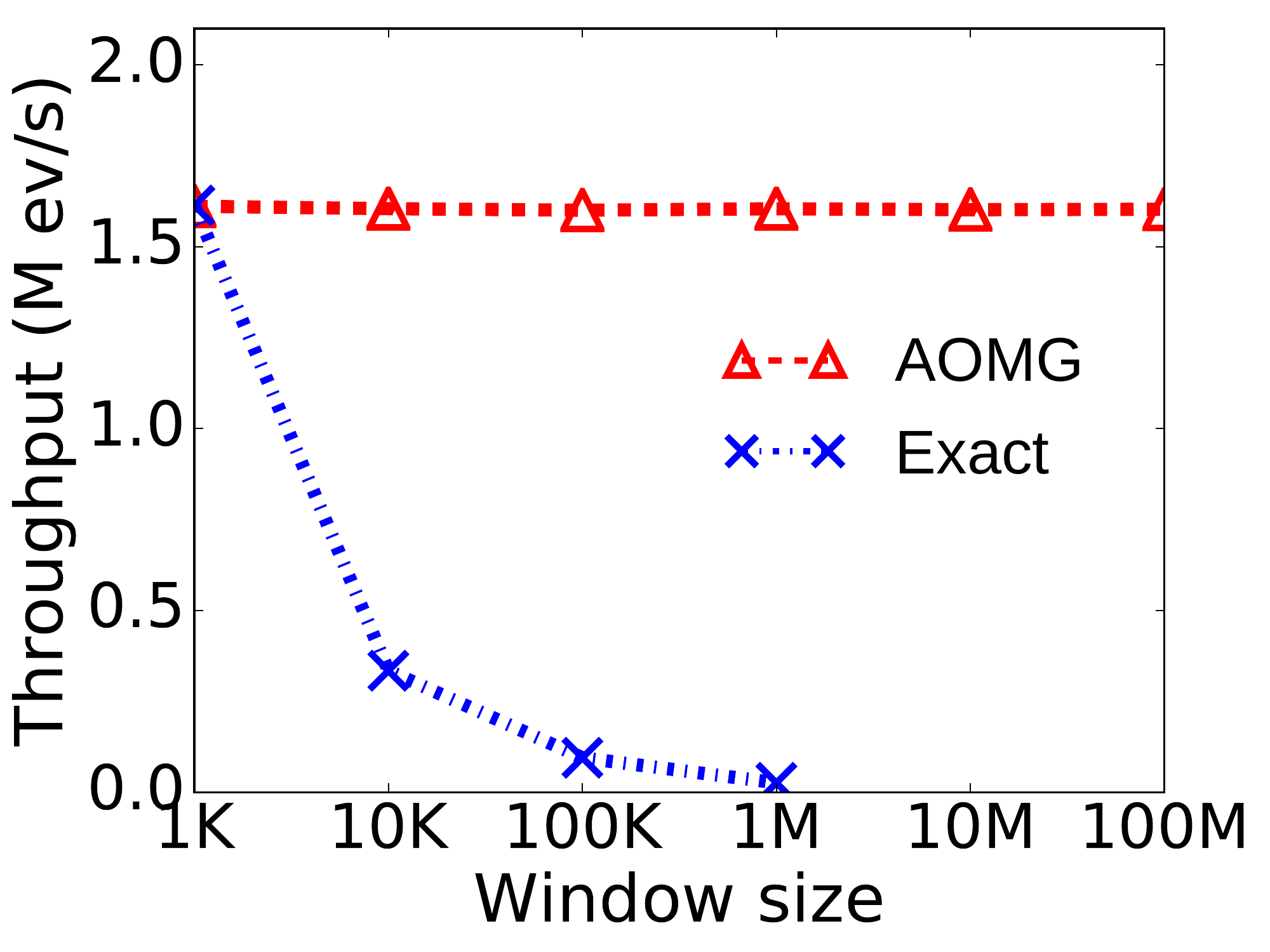}}
\subfigure[Uniform]{\includegraphics[height=1.25in]{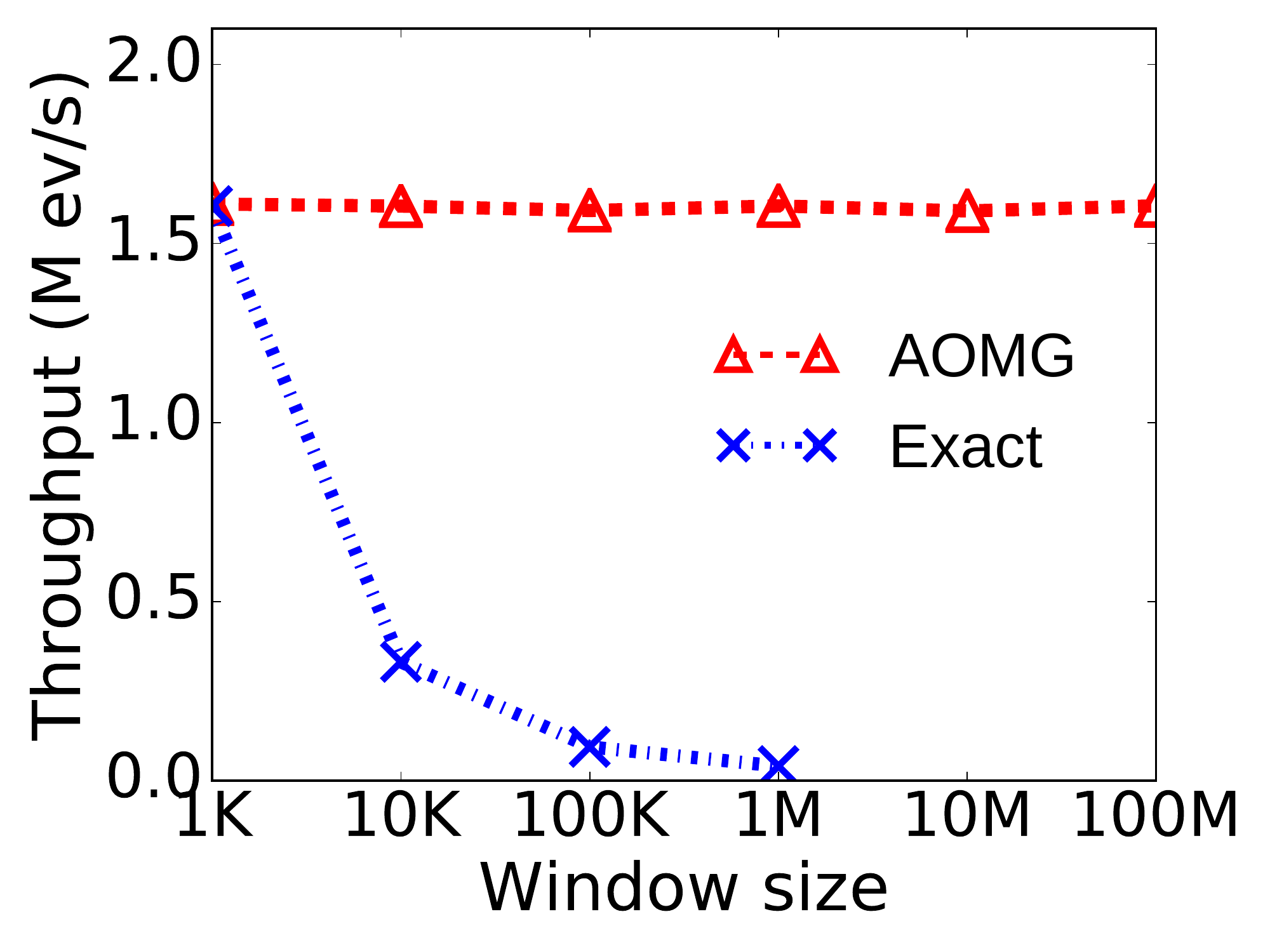}}
\end{adjustwidth}
\caption{\label{fig:throughput} Throughput of Exact and \both{} over varying window sizes. The window period is fixed as 1K elements.}
\vspace{-0.2pt}
\end{figure*}

Additionally, we test a larger $\epsilon = 0.1$ for CMQS, AM, and Random, and
$K = 3$ for Moment in order to reduce their space usage. It goes considerably
down to around 6,000, but value errors become extremely high.

\Paragraph{Throughput.} For throughput, we compare \both{} with CMQS, which
is observed as the most high-performance among rank-bound algorithms. In
CMQS, each sub-window creates a data structure, namely a \emph{sketch}, and
all active sketches are combined to compute approximate \qts{} over a sliding
window. The capacity of each sub-window is $\lfloor \frac{\epsilon N}{2}
\rfloor$ to ensure $\epsilon$-approximation under a sliding window~\cite{Lin04}.
In other word, if the sizes relevant
to a sliding window are given, we can get $\epsilon$ to be deterministically
ensured. In this experiment, we consider a query with 1K period and 100K
window. Here the $\epsilon$ is calculated as 0.02, and to cover a wider error
spectrum, we enforce the $\epsilon$ to range from 0.02 (1x) to 0.2 (10x).

Figure~\ref{fig:gk_throughput} presents the throughput of \both{} compared
with CMQS for varying $\epsilon$ values and also with Exact. Overall, \both{}
achieves higher throughput than CMQS across $\epsilon$ values and Exact.
CMQS has a clear trade-off between accuracy and throughput. If $\epsilon$ is
set too small (e.g., 1x), then the strategy will be too aggressive
and will largely lower performance (even lower than Exact). If $\epsilon$ is set
too large (e.g., 10x), then the throughput is largely recovered.
However, in this case, the strategy becomes too conservative and will be too
loose in bounding the error. In theory, this will allow a rank distance for
approximate \qt{} up to $\epsilon N = 20K$, which is unacceptable.



\Paragraph{Scalability.}  Figure~\ref{fig:throughput} presents the throughput
of \both{} and Exact with respect to different window sizes.
With window period fixed with 1K
elements, we vary the window size over a wide range from 1K to 1M elements
in x-axis, covering the use of tumbling window and sliding window that
contains up to 1K sub-windows. For the two synthetic datasets, Normal and Uniform,
we increase the window size up to 100M elements to further stress the query.


\both{} shows the consistent throughput for all window sizes for all
datasets. In comparison, Exact has throughput degradation as it begins to use
sliding window. For example, when the window size is increased to 10K, Exact shows
throughput degradation by 58-79\%. This is a consequence of paying \dacon{} cost
to search and eliminate the oldest 1K elements from the tree state for every
windowing period. \both{} achieves high scalability by mitigating such
\dacon{} cost and using small-size state as a summary of each sub-window.

We so far have presented how \both{} achieves low value error, low space
consumption, and scalable performance. Next, we present benefit of using
\topk{}.

\subsection{Few-k Merging}
\label{sec:FewkEval}


\begin{figure}[!t]
\centering
\includegraphics[width=2.4in]{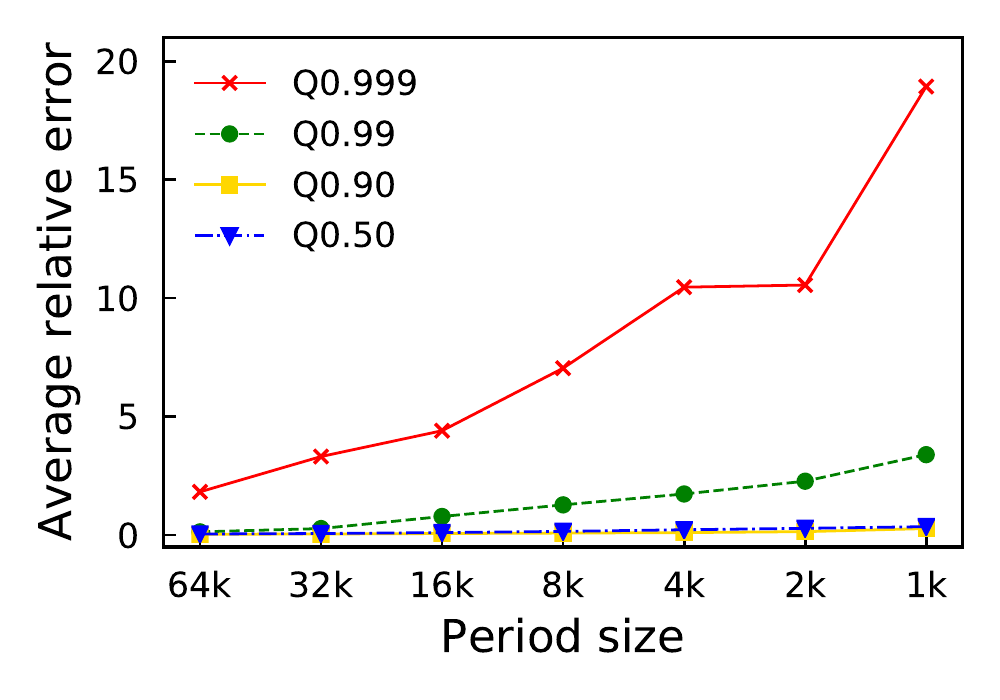}
\vspace{-0.1in}
\caption{\label{fig:relative_errors} Average relative errors without few-k merging (window size = 128K).}
\end{figure}

\Paragraph{Addressing statistical inefficiency.} As explained in
Section~\ref{sec:Limitations}, a larger period enables us to use more data
points to estimate high \qts{}, and deliver more accurate results.
Queries with small periods are where the top-k merging that
caches the largest values is
effective. To quantify this, we fix window size to include 128K elements, and
we vary the period size over a wide range from 64K to 1K on \pingmesh{}
dataset.

First, Figure~\ref{fig:relative_errors} summarizes average relative value
errors prior to applying the top-k merging on \pingmesh{}. We observe that varying period
sizes is insignificant to Q0.5 and Q0.9 with relative errors less than 1\%,
whereas it matters to Q0.999 with the error going up to 18.93\%. The accuracy
target is domain-specific.
For example, in our \pingmesh{} monitoring system in production,
$\approx$5\% of the relative
error is considered adequate. Therefore, if we set this as the optimization
target, the algorithm presented in Section~\ref{sec:ApproximateEvaluation}
alone is not sufficient.

\begin{table}[!t]
\begin{adjustwidth}{-1in}{-1in}
\centering
\scriptsize
{
\begin{tabular}{c|llll}
\hline
Fraction & 8K & 4K & 2K & 1K \\\hline
0.1 & 5.54 (209) & 2.43 (419) & 1.67 (838) & 1.30 (1,677) \\
0.5 & 0.68 (1,049) & 0.40 (2,097) & 0.36 (4,194) & 0.35 (8,389) \\\hline
\end{tabular}
}
\end{adjustwidth}
\caption{Average relative errors (and observed space usage) for using fraction in top-k merging w.r.t. the exact Q0.999.}
\label{tab:k-caching}
\vspace{-0.2in}
\end{table}

Having focused on Q0.999 in \pingmesh{}, we measure the accuracy by varying
the fraction of the caching size that \emph{guarantees} the exact answer in
the top-k merging, and show results in Table~\ref{tab:k-caching} along with the
observed space usage. Considering the 128K window size, each sub-window needs
to maintain $128K(1-0.999) = 132$ largest entries for the exact answer. As
the table shows, exploiting a much smaller fraction of it for each sub-window's
data can reduce Q0.999 value errors significantly. In particular, using a
half the space (\textit{i.e.}, fraction of 0.5) results in accuracy as close
as the optimal solution that needs the entire top-132 values. Using top-13
values (\textit{i.e.}, fraction of 0.1) makes the error fall around/below our
desired value-error target ($\approx$5\%). This excellent trade-off also
indicates that the largest entries in \pingmesh{} are fairly well distributed
across sub-windows.

Note that in \pingmesh{} missing the largest values in each sub-window indeed
hurts the accuracy. For example, we instead apply interval sampling
using a fraction 0.5, the value errors in 8K, 4K, 2K, and 1K periods explode
with 2.23\%, 4.60\%, 8.33\%, and 13.36\%, respectively.

\Paragraph{Addressing bursty tail.} Next, we discuss the effect of
the sample-k merging on bursty tail. The nature of bursty tail is that
the largest values from one or few sub-windows decide the target high \qt{}
in the window. Sampling here aims at picturing those values using a smaller
amount of space. To evaluate its effect,
we inject a bursty tail traffic to \pingmesh{} such
that it affects Q0.999 and above and appears just once in every evaluation of
the sliding window. In particular, for the window size $N$ and the \qt{} $\phi$, we
increase the values of the top $N(1-\phi)$ elements in every $(N/P)$th
sub-window of the window period $P$ by 10x.

Table~\ref{tab:k-sampling} presents average relative errors for two queries
with 16K and 4K period sizes, both using the 128K window. The fraction is
defined similarly: the amount of data assigned for holding sampled data
relative to the amount needed to give the guaranteed exact answer. The zero
fraction indicates the case that \both{} handles bursty tail without
the samples.
Looking at the zero fraction, the bursty tail is damaging
Q0.999 in both queries, and Q0.99 in the 4K-period query. This
is because burst tail blows up more when using smaller periods. That is,
the bursty tail exhibits the top-132 values which will sweep in the 40th
largest value that the 4K-period query refers for Q0.99.

Using the sampled values, both queries can improve their accuracy. In
general, Q0.999 needs a higher sampling rate (e.g., fraction of 0.5)
since the neighboring values are heavy-tailed. Q0.99 works well even with
conservative sampling (e.g., fraction of 0.1) since the neighboring
entries are gathered in a smaller value range.

\begin{table}[!t]
\begin{adjustwidth}{-1in}{-1in}
\centering
\scriptsize
{
\begin{tabular}{c|ll|ll}
\hline
\multirow{2}{*}{\bf Fraction} & \multicolumn{2}{c|}{\bf 16K} & \multicolumn{2}{c}{\bf 4K} \\
\cline{2-5}
& Q0.99 & Q0.999 & Q0.99 & Q0.999 \\
\hline\hline
0.0 & 0.08 (0) & 44.10 (0) & 28.15 (0) & 55.36 (0) \\
0.1 & 0.14 (1,048) & 25.97 (104) & 0.43 (4,194) & 17.38 (419) \\
0.5 & 0.05 (5,242) & 1.75 (524) & 0.30 (20,971) & 1.52 (2,097) \\
\hline
\end{tabular}
}
\end{adjustwidth}
\caption{\label{tab:k-sampling} Average relative errors (and observed space usage) for using fraction in sample-k merging w.r.t. the exact Q0.999.}
\vspace{-0.2in}
\end{table}

\Paragraph{Throughput.} Throughput in \topk{} is tightly coupled with the
number of entries to process per window. This is because the merged values
must be kept in a sorted form, and utilized directly.
With more merged data, the state grows bigger, consumes more processing
cycles, and thus affects throughput.
To illustrate, we evaluate the top-k merging on a large-window (100K)
small-period (1K) query,
which demands large space for high accuracy by the top-k merging.
With all entries cached
(i.e., fraction of 1), we see 21.2\% throughput penalty compared to \both{}
without the top-k merging. However, at a smaller fraction of 0.2,
where the average relative
error is only 0.6\%, throughput penalty is recovered to 9.0\%.

\subsection{Non-i.i.d. dataset}
\label{sec:Sensitivity}

We test if our aggregated estimator can be
applied to some non-i.i.d. data with competitive accuracy compared to i.i.d
data. To model a diverse spectrum of data dependence, we generate a
non-i.i.d. dataset from an AR(1) model, \textit{i.e.}, autoregressive
model~\cite{box15} of order 1, with coefficient $\psi \in \{0.1, 0.2, \cdots,
0.9\}$, where (1) $\psi$ represents the correlation between a data point and
its next data point, and (2) a larger $\psi$ indicates a stronger correlation
among neighboring data points. Data points in the dataset are identically and
normally distributed, with a mean of 1 million and a standard deviation of 50
thousand. For the purpose of comparison, we generate another i.i.d. dataset
from a normal distribution with the same mean and standard deviation, which
is equivalent to the AR(1) model with $\psi = 0$.

We evaluate the average relative errors between the
estimated and exact values for different \qts{}. Table~\ref{tab:iid} shows
the results for some selected \qts{} using three datasets that range from low
correlation to high correlation. We find that the errors slightly increase
when $\psi = 0.2$ (\textit{i.e.}, non-i.i.d. data with low correlation), and
mildly increase when $\psi = 0.8$ (\textit{i.e.}, non-i.i.d. data with high
correlation), compared to those when $\psi = 0$ (\textit{i.e.}, i.i.d. data).
Also, empirical probabilities that the value errors are within the
corresponding error bounds by Theorem~\ref{unique} in Appendix
are \emph{always} 1 for different $\psi \in \{0.1,
0.2, \cdots, 0.9\}$ and \qts{} $\phi \in \{0.1, 0.2, \cdots, 0.99\}$.
Therefore, we achieve (1) competitive results of non-i.i.d. data with respect
to high accuracy of estimated \qts{}, and (2) high probabilities of absolute
errors within error bounds. Hence, our approach is robust to the underlying
dependence in some sense.

\begin{table}[!t]
\begin{adjustwidth}{-1in}{-1in}
\centering
\scriptsize
{
\begin{tabular}{c||ccc}
\hline
& \multicolumn{3}{c}{\bf Quantiles}\\
$\psi$ & 0.5 & 0.9 & 0.99 \\ \hline\hline
0.0 & $3.46 \times 10^{-5}$ & $1.23 \times 10^{-4}$ & $8.88 \times 10^{-4}$ \\
0.2 & $3.47 \times 10^{-5}$ & $1.39 \times 10^{-4}$ & $9.84 \times 10^{-4}$ \\
0.8 & $5.66 \times 10^{-5}$ & $3.35 \times 10^{-4}$ & $1.56 \times 10^{-3}$ \\\hline
\end{tabular}
}
\end{adjustwidth}
\caption{Average relative errors for three datasets from autoregressive model with different correlation factors.}
\label{tab:iid}
\vspace{-0.2in}
\end{table}

\section{Discussion}
\label{sec:Discussion}


\Paragraph{Applicability.}
The properties of streaming workloads to which \both{} is effective are not limited
to datacenter telemetry monitoring. We show the potential of
\both{}'s applicability by using a high-volume geospatial 
IoT stream of taxi trip reports from New York City~\cite{NYCTaxi}.
On the stream dataset, we run the Qmonitor query (in Section~\ref{sec:ExperimentalSetup})
that continuously computes a set of \qts{} on trip distances, and
compare the accuracy of \both{}
against other methods in comparison.
The results show that for Q.50 and Q,90, rank-error approximate methods and AOMG all
deliver value errors within 5\%. But, similar to Pingmesh and Search, the rank-error
approximate mothods suffers from high value errors in high \qts{} such as Q.999.
Specifically, for Q.999 AOMG exhibits 1.79\% in value error while AM, Random, and Moment
exhibit 6.03\%, 44.43\%, and 21.77\% in value error, respectively. 
We observe that there is similarity in workload characteristics
between the taxi trip dataset and those described in Section~\ref{sec:Motivations}:
values in the middle of the distribution tend to be tightly clustered, and
distribution across sub windows is often self-similar.

There are also monitoring cases where 
most of approximate \qt{} policies exhibit high accuracy in value error
if rank error is small.
For instance, we deploy a query that monitors available host-side resources
such as CPU and memory on a cluster, and observe that
all Q0.999 estimates fall within 5\%.
This is mainly because the maximum value in the input stream is mostly small.
The server OS allocates memory as much as it can for buffer cache, and
this results in most of memory being consistently utilized, making available
memory mostly small.
Nonetheless, AOMG delivers higher throughput and scalability than other methods,
confirming its performance benefits.

\Paragraph{Limitations.}
The class of use cases we study is indeed
common and for those problems \both{} does achieve high accuracy and
low resource usage. However, if data distributions are changing drastically
in a short time frame and thus are not self-similar across sub windows,
\both{} (and rank-error approximate methods too) may not be effective. In this case, 
we could roll back to using the exact \qts{} and employ distributed monitoring
in larger scale to improve throughput.

\section{Related Work}
\label{sec:RelatedWork}

Stream processing engines have been developed for both
single-machine~\cite{Chandramouli14,Miao17} and distributed
computing~\cite{Qian13,Lin16} to provide runtime for parallelism,
scalability, and fault tolerance. \both{} can be applied to all these
frameworks. In the related work, we focus on \qt{} approximation algorithms
in the literature in details.

The work related to \qts{} approximation over data streams can be organized
into two main categories. The first one is a theoretical category that
focuses on enhancing space and time complexities. For the space complexity,
the space bound is a function of an error-tolerance parameter that the user
specifies~\cite{SelectionTimeBounds,Lin04,Arasu04,Shuzhuang17,Nikita19}.
The work in the second
category addresses challenges in \qts{} approximation under certain
assumptions, e.g., approximating \qts{} when the raw data is
distributed~\cite{Distributed_PODS_2012}, leveraging GPUs for \qt{}
computation~\cite{GPU_SIGMOD_2005}, utilizing both streaming data
and historical data stored on disk at the same time~\cite{Sneha16}.
The error bounds of the approximating
techniques in the two aforementioned categories are defined in terms of the
rank. In contrast, \both{} is driven by insights from workload
characteristics, and exploits diffeernt approaches taking into account
the underlying data distribution, producing low value errors.


In~\cite{SelectionTimeBounds}, linear time algorithms were presented to
compute a given \qt{} exactly using a fixed set of elements,
while~\cite{Lin04,Arasu04} presented algorithms for \qt{} approximation over
sliding windows. The work in~\cite{Lin04,Arasu04} is the most related work to
\both{}. In ~\cite{Arasu04}, the space complexity of~\cite{Lin04} is
improved: \textit{i.e.}, less memory space is used to achieve the same
accuracy. Similarly, the randomized algorithms in \cite{Luo16} use the memory
size as a parameter to provide a desired error bound. However, the work in
~\cite{Lin04,Arasu04,Luo16} cannot scale for very large sliding windows when
low latency is a requirement. This is mainly because the cost for deaccumulating the
expired elements is not scaling to sliding-window size. In contrast, \both{}
can scale for large sliding windows due to its ability to deaccumulate an
entire expiring sub-window at a time with low cost. Hence, \both{} fits more
for real-time applications, where low latency is a key requirement.

Many research
efforts~\cite{Distributed_PODS_2004,Distributed_SIGMOD_2011,Distributed_PODS_2012}
assume that the input data used to compute the \qts{} is distributed.
Similarly, the work done in
~\cite{Continuous_SIGMOD_2005,Continuous_PODS_2009} computes \qts{} over
distributed data and takes a further step by continuously monitoring any
updates to maintain the latest \qts{}.
In \both{}, we target applications where a large stream of data may originate
from different sources to be processed by a streaming engine. \both{}
performs a single pass over the data to scale for large volume of input
streams.

Deterministic algorithms for approximating biased \qts{} are first presented
in~\cite{Biased_PODS_2006}. Biased \qts{} are those with extreme values,
e.g., 0.99-\qt{}. In contrast, \both{} is designed to compute both
biased and unbiased \qts{}. Moreover,~\cite{Biased_PODS_2006} is sensitive to
the maximum value of the streaming elements, while \both{} is not. In
particular, the memory consumed by~\cite{Biased_PODS_2006} includes a
parameter
that represents the maximum value a streaming element can have. Biased \qts{}
can have very large values in many applications. \both{} is able to estimate
them without any cost associated with the actual values of the biased \qts{}.

\section{Conclusion}
\label{sec:Conclusion}

\Qts{} is a challenging operator in real-time streaming analytics systems as
it requires high throughput, low latency, and high accuracy. We present
\both{} that satisfies these requirements through workload-driven and
value-based \qt{} approximation. We evaluated \both{} using synthetic and
real-world workloads on a state-of-the-art streaming engine demonstrating
high throughput over a wide range of window sizes while delivering small
relative value error. Although the evaluation is based on single machine, our
\qt{} design can deliver better aggregate throughput while using a fewer
number of machines in distributed computing.


{\footnotesize
\bibliographystyle{plain}
\bibliography{references}
}

\appendix

\begin{theorem}{\label{unique}}
Suppose the data in the sliding window are independent and identically distributed (i.i.d.). When $m \rightarrow \infty$, with probability at least $1 - \alpha$, the following holds asymptotically
\[ |y_a - y_e| \leq \frac{2\Phi^{-1}(\alpha/2)\sqrt{\phi(1-\phi)}}{\sqrt{nm} f(p_\phi)} \]
, where $\Phi^{-1}(\alpha)$ is the upper $\alpha$-\qt{} of standard normal distribution and its inverse $\Phi(\cdot)$ is the cumulative distribution function of standard normal distribution, $f(\cdot)$ is the probability density of the data distribution at its $\phi$-\qt{} $p_\phi$.

In particular, we take $\alpha = 5\%$. When $m \rightarrow \infty$, with probability at least $95\%$, the following holds asymptotically
\[ |y_a - y_e| \leq \frac{2 \times 1.96 \sqrt{\phi(1-\phi)}}{\sqrt{nm} f(p_\phi)} \]
\end{theorem}

\begin{proof}
By the Central Limit Theorem for the sample $\phi$-\qt{} of \emph{i.i.d.} data \cite{stuart1994kendall, tierney1983space} in each sub-window, we have
\begin{equation*}
y_i \sim \mathcal{N}(p_\phi, \frac{\phi(1-\phi)}{mf(p_\phi)^2})
\end{equation*}
when $m \rightarrow \infty$.
Since data $x_{i,j}, i = 1, \cdots, n, j = 1, \cdots, m$ are $i.i.d.$, $y_i, i = 1, \cdots, n$ are \emph{i.i.d.} as well. Then for the aggregated estimate, we have
\begin{equation*}
y_a = \frac{1}{n} \sum_{i = 1}^n y_i \sim \mathcal{N}(p_\phi, \frac{\phi(1-\phi)}{nmf(p_\phi)^2})
\end{equation*}
when $m \rightarrow \infty$. Therefore, with probability $1 - \alpha/2$, the following holds asymptotically
\begin{equation*}
|y_a - p_\phi| \leq \frac{\Phi^{-1}(\alpha/2)\sqrt{\phi(1-\phi)}}{\sqrt{nm} f(p_\phi)}
\end{equation*}
when $m \rightarrow \infty$. On the other hand, for the sample $\phi$-\qt{} of the sliding window, we have
\begin{equation*}
y_e \sim \mathcal{N}(p_\phi, \frac{\phi(1-\phi)}{nmf(p_\phi)^2})
\end{equation*}
when $m \rightarrow \infty$. Therefore, with probability $1 - \alpha/2$, the following holds asymptotically
\begin{equation*}
|y_e - p_\phi| \leq \frac{\Phi^{-1}(\alpha/2)\sqrt{\phi(1-\phi)}}{\sqrt{nm} f(p_\phi)}
\end{equation*}
when $m \rightarrow \infty$. Combining (2) and (3), with probability at least $1 - \alpha$, the following holds asymptotically when $m \rightarrow \infty$.
\begin{equation*}
|y_e - y_a| \leq \frac{2\Phi^{-1}(\alpha/2)\sqrt{\phi(1-\phi)}}{\sqrt{nm} f(p_\phi)}
\end{equation*}
\end{proof}

\end{document}